\documentclass[preprint,12pt]{elsarticle}

\usepackage{amsmath,amssymb,amsfonts}
\usepackage{amsthm}
\usepackage{graphicx}
\usepackage{booktabs}
\usepackage{url}
\usepackage[colorlinks=true,allcolors=blue]{hyperref}
\usepackage{bm}
\usepackage{enumitem}

\newtheorem{theorem}{Theorem}

\newtheorem{proposition}{Proposition}

\newtheorem{definition}{Definition}
\newtheorem{remark}{Remark}

\newtheorem{observation}[theorem]{Empirical observation}

\newcommand{\Z}{\mathbb{Z}}
\newcommand{\C}{\mathbb{C}}
\newcommand{\R}{\mathbb{R}}
\newcommand{\Cay}{\mathrm{Cay}}
\newcommand{\Aut}{\mathrm{Aut}}
\newcommand{\dG}{d_{G}}
\newcommand{\dGam}{d_{\Gamma}}

\newcommand{\F}{\mathbb{F}}
\newcommand{\Lift}{\mathcal{L}}
\newcommand{\Restr}{\mathcal{R}}
\newcommand{\Id}{\mathrm{Id}}
\DeclareMathOperator{\diam}{diam}

\journal{arXiv Preprint}

\begin{document}
	\begin{frontmatter}
	
	\title{Tight Wavelet Frames on Graphs via\\ Isometric Group Embedding}
	
	\author[inst1]{Rigobert~Fokam~Souop\corref{cor1}}
	\ead{fokamrigobert@gmail.com}
	\cortext[cor1]{Corresponding author}
	
	\author[inst1]{Laurent~Bitjoka}
	
	\affiliation[inst1]{organization={Laboratory of Energy, Signal, Imaging and Automation (LESIA), Department of Electrical Engineering, Energetics and Automation},
		addressline={University of Ngaound\'er\'e}, 
		city={Ngaound\'er\'e},
		country={Cameroon}}
	
	\begin{abstract}
		Spectral graph wavelets apply a kernel to the graph Laplacian spectrum. On an
		irregular graph their analyzing functions inherit a non-canonical eigenbasis,
		they do not form a tight frame, and reconstruction requires inverting a frame
		operator. We take a different route, built on an exact substrate. Given an
		isometric embedding of a connected graph into a Cayley graph of a finite
		abelian group---a host on which classical Fourier analysis applies exactly---we
		construct wavelets on the host and restrict them to the graph. Two
		constructions arise and we keep them separate. \emph{Dilation wavelets} use a
		group automorphism as a dilation, reproducing the classical translate--dilate
		template but existing only on hosts with composite cyclic factors.
		\emph{Spectral band-pass wavelets} use a normalized filter bank in the dual
		frequency magnitude; they exist on every host, form a Parseval (tight) frame,
		reconstruct any graph signal exactly via restriction, are
		translation-covariant, and localize jointly in vertex and frequency. We prove
		the tight-frame identity and exact reconstruction, give a multiresolution
		decomposition, and show the full transform costs $O(JN\log N)$ via the host
		fast Fourier transform. For a proper embedding we show the canonical way to
		complete a signal onto the host remainder is the discrete \emph{harmonic
			extension}, which uniquely minimizes the host Dirichlet energy and places the
		zero-padding and symmetric-extension heuristics as approximations of it. On
		benchmark hosts reconstruction reaches machine precision ($\sim\!10^{-15}$) and
		band-pass atoms concentrate $89$--$99\%$ of their energy within graph-distance
		two of their center.
	\end{abstract}
	

	\end{frontmatter}

	\section{Introduction}
	\label{sec:intro}
	
	The extension of wavelet analysis from regular domains to graphs is by now
	a mature subject. The spectral graph wavelet transform of Hammond,
	Vandergheynst, and Gribonval~\cite{Hammond2011,HammondVandergheynstGribonval2019}
	defines wavelets through a kernel $g$ evaluated on the eigenvalues of the
	graph Laplacian, and remains the reference construction against which we
	measure ours. Earlier and parallel approaches reach related goals by other
	means: diffusion wavelets built from powers of a diffusion
	operator~\cite{CoifmanMaggioni2006}, spatial graph wavelets for traffic
	analysis~\cite{CrovellaKolaczyk2003}, and critically sampled two-channel
	filter banks on bipartite graphs~\cite{NarangOrtega2012}. A complementary
	line uses wavelets as a tool for graph analysis rather than signal
	reconstruction, for instance in multiscale community
	mining~\cite{TremblayBorgnat2014}.
	
	A substantial literature pursues \emph{tight} frames directly on the graph,
	typically by shaping the spectral kernel rather than changing the underlying
	eigenbasis~\cite{LeonardiVanDeVille2013,ShumanWiesmeyrHolighausVandergheynst2015,BehjatRichterVanDeVilleSornmo2016}.
	Subsequent designs trade exact tightness for polynomial (hence localized and
	fast) filters~\cite{TayTanakaSakiyama2017}, develop sparse representations
	and Parseval frames adapted to particular graph
	families~\cite{Dong2017,NakahiraMiyamoto2018}, or analyse such frames
	statistically for denoising~\cite{GobelBlanchardVonLuxburg2018}. Related
	frame-theoretic constructions target Gabor-type atoms and ranked-data domains
	with intrinsic group
	structure~\cite{GhandehariGuillotHollingsworth2021,ChenDeJongHalversonShuman2021},
	while applications to neuroimaging~\cite{BehjatLeonardiSornmoVanDeVille2015}
	and learned constructions~\cite{RustamovGuibas2013} indicate the breadth of
	the program. What these otherwise diverse constructions share is a single
	foundation, common also to the broader theory of signal processing on
	graphs~\cite{Shuman2013,Ortega2018}: a shift operator---the Laplacian, or
	the adjacency-based variant of discrete signal processing on
	graphs~\cite{SandryhailaMoura2013}---is diagonalized, and its eigenvectors
	are taken as Fourier modes. The structural cost of that choice, and the
	alternative we propose, are the subject of this paper.
	
	That foundation carries a structural cost. The eigenvectors of a generic
	graph operator are not canonical---an eigenvalue of multiplicity $m$ fixes
	only an $m$-dimensional eigenspace, not a basis---so the resulting transform
	is defined only up to rotations within degenerate eigenspaces. There is no
	shift that acts as a genuine translation, the analyzing frame is in general
	not tight, and recovering a signal from its wavelet coefficients requires
	solving a linear system with the frame operator \cite{Hammond2011}.
	
	This paper develops graph wavelets on a different substrate. Suppose the graph
	$G$ is embedded \emph{isometrically} into a Cayley graph of a finite abelian
	group $\Gamma$; such an embedding always exists, and can be made compact, by
	the construction recalled in Section~\ref{sec:prelim} (with full proofs in the
	companion work \cite{FokamThesis2026}). On the host $\Gamma$ the characters
	furnish a \emph{canonical} orthonormal Fourier basis, translation is the group
	action, and convolution is exact \cite{ShiMoura2019}. We build wavelets there
	and pull them back
	to $G$. The payoff is a wavelet system whose mathematical properties are those
	of classical wavelet analysis \cite{Mallat2009,Daubechies1992,Meyer1992}
	rather than approximations of it.
	
	Our contributions are the following.
	\begin{enumerate}[label=(\roman*)]
		\item We separate two genuinely different constructions that the literature on
		``group'' graph wavelets tends to conflate
		(Section~\ref{sec:tightframe}). \emph{Dilation wavelets}
		require a dilation automorphism and exist only on hosts with composite
		cyclic factors; on a binary host $\Z_2^k$ no nontrivial scalar dilation
		exists, so this construction is genuinely restricted. \emph{Spectral
			band-pass wavelets} exist on every host, and are the construction we
		develop.
		\item We prove that the normalized spectral filter bank yields a Parseval
		tight frame and that every graph signal is reconstructed exactly, by
		restriction, from its wavelet coefficients (Theorem~\ref{thm:tightframe}).
		\item We establish translation covariance and joint localization
		(Section~\ref{sec:localization}), give a multiresolution decomposition
		(Section~\ref{sec:mra}), and an $O(JN\log N)$ fast transform
		(Section~\ref{sec:complexity}).
		\item For a proper embedding we identify the canonical completion of a graph
		signal onto the host complement: the discrete harmonic extension, which
		uniquely minimizes host Dirichlet energy and subsumes the zero-padding and
		symmetric-extension heuristics as approximations
		(Section~\ref{sec:completion}, Theorem~\ref{thm:harmonic}).
		\item We validate the theory numerically (Section~\ref{sec:experiments}):
		machine-precision reconstruction, measured localization, the optimality of
		the harmonic completion, and a comparison
		with the spectral graph wavelet transform that is explicit about where each
		approach is preferable (Section~\ref{sec:discussion}).
	\end{enumerate}
	
	\paragraph{Naming.}
	Following the companion papers~\cite{FokamP1,FokamP2}, we refer to the
	overall framework as \emph{Group Embedding-based Graph Signal Processing}
	(GE-GSP). Its Fourier layer---the group graph Fourier transform of
	Section~\ref{sec:prelim}---is the GE-GFT, and the wavelet transform
	constructed in this paper is the GE-GWT. We use these names throughout, in
	particular in the head-to-head of Section~\ref{sec:discussion} and
	Figure~\ref{fig:headtohead}.
	
	\paragraph{The methodological principle.}
	The construction rests on a single principle from classical harmonic analysis,
	used here as the mathematical foundation of the method rather than as an end in
	itself. The principle is that the characters of a finite abelian group furnish a
	\emph{canonical} orthonormal Fourier basis---fixed, with no eigenspace
	ambiguity---on which translation is exact modulation and convolution is exact,
	the finite-group specialization of Fourier analysis on locally compact abelian
	groups~\cite{Rudin1962,Folland2016}. Transporting an arbitrary graph onto such
	a group by an isometric embedding makes this apparatus available to it, and two
	concrete signal-processing objects result. First, the tight-frame identity
	(Theorem~\ref{thm:tightframe}) is a Parseval relation for a partition of unity
	in the character variable, exact rather than approximate because the basis is
	canonical---so reconstruction needs no frame-operator inversion. Second, the
	optimal host completion (Theorem~\ref{thm:harmonic}) is a discrete Dirichlet
	minimization: among all extensions of a graph signal to the host, the harmonic
	one minimizes Dirichlet energy and is computable through the same group Fourier
	transform. The excursion ratio $\varepsilon$ measures, in a single number, how
	far a given graph is from the classical abelian setting, and thus predicts when
	the method is fast and sharply localized. The value of the framework is
	methodological: it delivers, for any graph carrying abelian structure, the exact
	tightness, canonical basis, and true translation that a spectral construction on
	a non-canonical eigenbasis cannot.

	\section{Preliminaries: isometric group embedding and the group Fourier transform}
	\label{sec:prelim}
	
	We recall only what the wavelet construction needs; the embedding theory,
	including existence, compactness, and the algorithm that produces the host,
	is developed in the companion papers \cite{FokamP1,FokamP2}. Throughout, $G=(V,E)$ is a finite
	connected graph with shortest-path metric $\dG$, and $\Gamma$ is a finite
	abelian group.
	
	\begin{definition}[Cayley graph]
		For a generating set $S\subseteq\Gamma\setminus\{0\}$ with $S=-S$, the Cayley
		graph $\Cay(\Gamma,S)$ has vertex set $\Gamma$ and edges $\{g,g+s\}$ for
		$g\in\Gamma$, $s\in S$. It is vertex-transitive \cite{Godsil2001}, and its
		graph metric is
		translation invariant: $\dGam(x,y)=\dGam(0,y-x)$.
	\end{definition}
	
	\begin{definition}[Isometric embedding and lift]
		\label{def:embedding}
		An \emph{isometric embedding} of $G$ is an injection
		$\phi\colon V\to\Gamma$ into some $\Cay(\Gamma,S)$ with
		$\dGam(\phi(u),\phi(v))=\dG(u,v)$ for all $u,v\in V$. We write $N=|\Gamma|$ and
		call $\varepsilon=|V|/N\in(0,1]$ the \emph{excursion ratio}. The \emph{lift}
		$\Lift\colon\C^{V}\to\C^{\Gamma}$ sends $s$ to $\tilde s$ with
		$\tilde s(\phi(v))=s(v)$ and $\tilde s(g)=0$ for $g\notin\phi(V)$; the
		\emph{restriction} $\Restr\colon\C^\Gamma\to\C^V$ is
		$(\Restr f)(v)=f(\phi(v))$. Then $\Restr\Lift=\Id$.
	\end{definition}
	
	The existence of such an embedding for every connected graph, together with a
	construction minimizing $N$ heuristically, is the main result of the
	companion papers \cite{FokamP1,FokamP2} (with full detail in \cite{FokamThesis2026}). When $G$ is itself an abelian Cayley graph (a cycle,
	torus, circulant, circular ladder), the embedding is onto, $\varepsilon=1$,
	and the theory below reduces to classical Euclidean signal processing.
	
	\begin{definition}[Characters and the group Fourier transform]
		\label{def:gft}
		The dual $\widehat\Gamma$ consists of the characters
		$\chi_k\colon\Gamma\to\C$, $k\in\widehat\Gamma\cong\Gamma$, which are
		orthonormal:
		$\frac1N\sum_{g}\chi_k(g)\overline{\chi_\ell(g)}=\delta_{k\ell}$. The
		\emph{group graph Fourier transform} (GFT) of $s\in\C^V$ is
		\[
		\hat s(k)=\frac{1}{\sqrt N}\sum_{g\in\Gamma}\tilde s(g)\,\overline{\chi_k(g)},
		\qquad
		\tilde s(g)=\frac{1}{\sqrt N}\sum_{k}\hat s(k)\,\chi_k(g).
		\]
		The matrix $F_{k,g}=\chi_k(g)/\sqrt N$ is unitary, so the GFT preserves the
		$\ell^2$ norm of lifted signals (Plancherel) --- the classical harmonic
		analysis of a finite abelian group \cite{Terras1999,Rudin1962}, specialized
		to the present setting.
	\end{definition}
	
	For a product host $\Gamma=\Z_{N_1}\times\cdots\times\Z_{N_d}$ the characters
	are products of roots of unity,
	$\chi_k(g)=\prod_{r}\exp(2\pi i\,k_r g_r/N_r)$, and the GFT is a mixed-radix
	fast Fourier transform. We write $|k|$ for the \emph{frequency magnitude} of
	$\chi_k$, the word length of $k$ on the dual generating set; concretely
	$|k|=\bigl(\sum_r\min(k_r,N_r-k_r)^2\bigr)^{1/2}$, which equals $\dGam(0,\cdot)$
	transported to $\widehat\Gamma$ and orders characters from smooth (small $|k|$)
	to oscillatory (large $|k|$).
	
	\begin{definition}[Translation]
		\label{def:translation}
		For $h\in\Gamma$ the translation $T_h$ acts on lifts by
		$(T_h\tilde s)(g)=\tilde s(g-h)$; equivalently
		$\widehat{T_h\tilde s}(k)=\overline{\chi_k(h)}\,\hat s(k)$. Each $T_h$ is
		unitary, and $T_gT_h=T_{g+h}$. Because $\Gamma$ is abelian its characters are
		its irreducible representations, each one-dimensional, and the multiplicative
		identity $\chi_k(g+h)=\chi_k(g)\chi_k(h)$ is exactly what makes translation
		act as pure modulation $\hat s(k)\mapsto\overline{\chi_k(h)}\,\hat s(k)$ in the
		dual~\cite{Terras1999,Folland2016}. This is what fails for the eigenvector
		``Fourier'' transform of a generic graph, whose shift operator is not a group
		action and whose frequency analysis is correspondingly
		basis-dependent~\cite{SandryhailaMoura2014}.
	\end{definition}
	
	Two facts will be used repeatedly: the GFT is unitary on the lifted subspace
	$\Lift(\C^V)\subseteq\C^\Gamma$ (Plancherel), and translation is exact on the
	host. Both fail for the eigenvector ``Fourier'' transform of a generic graph;
	they hold here because the host is a group and $\phi$ is isometric, so that
	host distance equals graph distance and a character of small magnitude is a
	slowly varying function along the graph.
	
	\section{Spectral band-pass wavelets: a universal tight frame}
	\label{sec:tightframe}
	\label{sec:spectral}
	
	We now give the construction that exists on every host. It mirrors the
	spectral graph wavelet transform \cite{Hammond2011}, but on the canonical
	characters rather than Laplacian eigenvectors, and it is engineered to be
	tight.
	
	\begin{definition}[Wavelet filter bank]
		\label{def:filterbank}
		Fix a low-pass (scaling) kernel $g_0$ and $J$ band-pass kernels
		$g_1,\dots,g_J\colon[0,|k|_{\max}]\to[0,\infty)$. The \emph{normalized filter
			bank} is
		\[
		\hat\psi_j(k)=\frac{g_j(|k|)}{\sqrt{\sum_{i=0}^{J}g_i(|k|)^2}},
		\qquad j=0,1,\dots,J,
		\]
		so that $\sum_{j=0}^{J}|\hat\psi_j(k)|^2=1$ for every $k$. The \emph{wavelet at
			scale $j$, centered at $h\in\Gamma$} is $\psi_{j,h}=T_h\mathcal F^{-1}\hat\psi_j$,
		that is $\widehat{\psi_{j,h}}(k)=\overline{\chi_k(h)}\,\hat\psi_j(k)$.
	\end{definition}
	
	The normalization is the whole point: it makes the bank a partition of the
	frequency axis in the $\ell^2$ sense, which is exactly the tight-frame
	condition.
	
	\begin{theorem}[Parseval tight frame and exact reconstruction]
		\label{thm:tightframe}
		The family $\{\psi_{j,h}:0\le j\le J,\ h\in\Gamma\}$ is a \emph{Parseval
			frame} \cite{Christensen2016} for
		$\C^\Gamma$: for every $\tilde s$,
		\[
		\sum_{j=0}^{J}\sum_{h\in\Gamma}|\langle\tilde s,\psi_{j,h}\rangle|^2
		=\|\tilde s\|_2^2,
		\qquad
		\tilde s=\sum_{j=0}^{J}\sum_{h\in\Gamma}\langle\tilde s,\psi_{j,h}\rangle\,
		\psi_{j,h}.
		\]
		Consequently every graph signal $s$ satisfies
		$s=\Restr\bigl(\sum_{j,h}W_s(j,h)\,\psi_{j,h}\bigr)$ with
		$W_s(j,h)=\langle\tilde s,\psi_{j,h}\rangle$; reconstruction is exact and uses
		no matrix inversion.
	\end{theorem}
	
	\begin{proof}
		The coefficient $W_s(j,h)=\langle\tilde s,\psi_{j,h}\rangle$ equals
		$\bigl(\overline{\hat\psi_j}\cdot\hat{\tilde s}\bigr)^{\vee}(h)$, the inverse
		transform of $\overline{\hat\psi_j}\,\hat{\tilde s}$ evaluated at $h$. By
		Plancherel in the $h$-summation, for each $j$,
		\[
		\sum_{h}|W_s(j,h)|^2=\sum_{k}|\hat\psi_j(k)|^2\,|\hat{\tilde s}(k)|^2 .
		\]
		Summing over $j$ and using $\sum_j|\hat\psi_j(k)|^2=1$
		(Definition~\ref{def:filterbank}),
		\[
		\sum_{j,h}|W_s(j,h)|^2=\sum_k|\hat{\tilde s}(k)|^2=\|\tilde s\|_2^2,
		\]
		which is the Parseval identity. For synthesis, the frame operator
		$S\tilde s=\sum_{j,h}\langle\tilde s,\psi_{j,h}\rangle\psi_{j,h}$ acts on the
		Fourier side as multiplication by $\sum_j|\hat\psi_j(k)|^2=1$, i.e.\ $S=\Id$;
		thus $\tilde s=\sum_{j,h}W_s(j,h)\psi_{j,h}$. Restricting and using
		$\Restr\Lift=\Id$ gives the statement for $s$.
	\end{proof}
	
	\begin{remark}
		The contrast with the spectral graph wavelet transform is sharp and
		structural. There the analyzing functions are $g(t_j\lambda)$ evaluated on
		Laplacian eigenvalues; the resulting frame is generally not tight, so
		reconstruction proceeds by inverting the frame operator (in practice by
		conjugate gradients) \cite{Hammond2011}. Tightening the spectral construction
		itself has also been pursued, by a different route (the eigenvector domain)
		than ours, through multislice and spectrum-adapted designs
		\cite{LeonardiVanDeVille2013,ShumanWiesmeyrHolighausVandergheynst2015}. Here
		the normalization builds the
		tight-frame condition in by construction, and synthesis is the adjoint of
		analysis with no inversion. The price is that we require a group host; the
		benefit is exactness.
	\end{remark}
	
	\begin{figure}[t]
		\centering
		\includegraphics[width=0.66\linewidth]{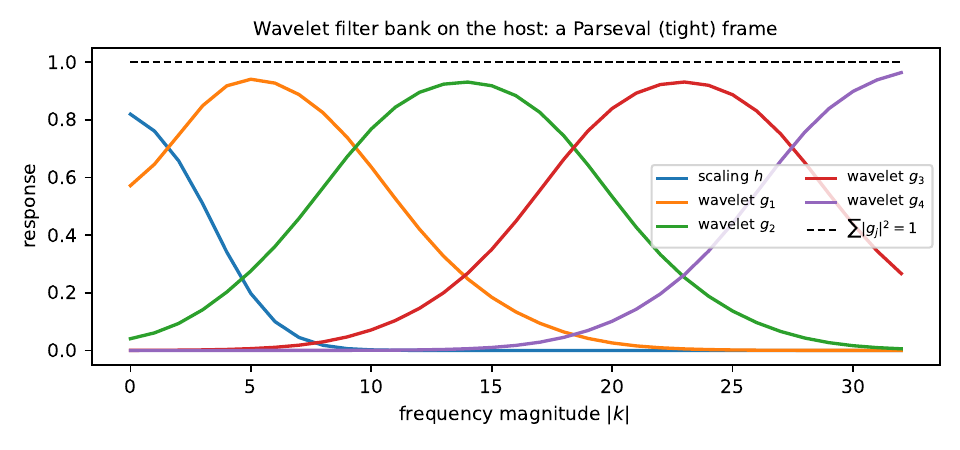}
		\caption{A normalized wavelet filter bank on the host $\Z_{64}$: one low-pass
			scaling kernel and $J=4$ band-pass kernels in the frequency magnitude $|k|$,
			with $\sum_j|\hat\psi_j(k)|^2\equiv1$ (dashed), the Parseval condition of
			Theorem~\ref{thm:tightframe}.}
		\label{fig:filterbank}
	\end{figure}
	
	\subsection{Relation to the classical translate--dilate template}
	\label{sec:dilation}
	
	The band-pass construction departs from the classical wavelet template in
	using a frequency-domain filter bank rather than dilations of a single mother
	wavelet. The literal transport of the translate--dilate template is also
	available, but only in restricted scope, and we record it briefly. A group
	automorphism $\alpha\in\Aut(\Gamma)$ acts as a dilation
	$(D_\alpha\tilde s)(g)=\tilde s(\alpha^{-1}g)$; it is unitary, satisfies
	$D_\alpha D_\beta=D_{\alpha\beta}$, and transforms as
	$\widehat{D_\alpha\tilde s}(k)=\hat s(\alpha^{\!\top}k)$ with no Jacobian
	factor, since $\alpha$ permutes $\Gamma$ exactly---unlike dilations on $\R^d$.
	With a mother wavelet $\psi$, a set $\mathcal A\subseteq\Aut(\Gamma)$, and
	$\psi_{\alpha,h}=D_\alpha T_h\psi$, the family
	$\{\psi_{\alpha,h}\}$ reconstructs every signal whenever
	$C_\psi(k)=\sum_{\alpha\in\mathcal A}|\hat\psi(\alpha^{\!\top}k)|^2$ is
	positive for all $k\neq0$, by the same Plancherel argument as
	Theorem~\ref{thm:tightframe}.
	
	The decisive limitation is that scalar dilations exist only on hosts with
	composite cyclic factors. On $\Z_N$ the dilations are multiplication by the
	$\varphi(N)$ units: the ring $C_{16}$ admits eight, the circulant
	$C_{12}(1,2)$ four. On a binary host $\Z_2^k$ the only automorphism fixing the
	generators as scalars is the identity, so no nontrivial dilation exists and the
	construction is vacuous. Because the compact hosts of generic graphs are
	frequently binary for graphs lacking abelian symmetry~\cite{FokamP2}, the dilation wavelets cannot be the
	universal tool; that role is played by the spectral band-pass frame above,
	which exists on every host. We therefore develop the band-pass construction in
	the remainder and treat dilation wavelets only as the bridge to the classical
	picture on cyclic hosts.
	
	\section{Localization}
	\label{sec:localization}
	
	A wavelet must localize jointly in vertex and frequency
	\cite{StankovicSejdicDakovic2018}. On the host this is
	exact in frequency and provably covariant in space; restricted to the graph,
	spatial concentration is a property of the embedding that we then measure.
	
	\begin{proposition}[Frequency localization and covariance]
		\label{prop:freqloc}
		Each $\psi_{j,h}$ has Fourier transform supported where $g_j$ is nonzero, and
		the family is translation covariant: $\psi_{j,h}=T_h\psi_{j,0}$. Consequently
		$W_s(j,\cdot)=\tilde s*\overline{\psi_{j,0}^-}$ is a convolution (a filtering)
		followed by sampling on $\phi(V)$, where $\psi^-(g)=\psi(-g)$.
	\end{proposition}
	
	\begin{proof}
		Frequency support is read from
		$\widehat{\psi_{j,h}}(k)=\overline{\chi_k(h)}\hat\psi_j(k)$. Covariance is the
		modulation--translation duality of Definition~\ref{def:translation}, and the
		convolution form is the standard rewriting of analysis against a translated
		family.
	\end{proof}
	
	Spatial localization does not hold on an arbitrary host: an arbitrary
	combination of band-limited characters need not concentrate near a vertex. It
	holds here for a specific reason---the embedding is isometric, so host distance
	from $h$ equals graph distance, and a frequency-band-limited function on a
	cyclic or toral host is spatially concentrated by the discrete uncertainty
	principle \cite{DonohoStark1989,PerraudinRicaudShumanVandergheynst2018}. We
	quantify the effect.
	
	\begin{observation}[Measured spatial localization]
		\label{prop:spatialloc}
		For a band-pass atom $\psi_{j,0}$ centered at a vertex $v_0$, define the
		concentration $\rho_j=\sum_{v:\dG(v,v_0)\le2}|\psi_{j,0}(\phi(v))|^2/
		\|\psi_{j,0}\|^2$. With Gaussian band-pass kernels and $J=4$, the computed
		values are $\rho=(0.95,0.97,0.89,0.95)$ on the ring $C_{16}$ (host $\Z_{16}$)
		and $\rho=(0.99,0.98,0.98,0.94)$ on the grid $6\times6$ (host $\Z_{10}^2$);
		finer scales concentrate more tightly. (This measure presumes diameter greater
		than two; see Section~\ref{sec:exp-proper} for low-diameter hosts.)
	\end{observation}
	
	\begin{remark}[Where localization weakens]
		\label{rem:binaryground}
		On a binary host $\Z_2^k$ with large $k$ the same kernels still produce a
		tight frame and exact reconstruction (Theorem~\ref{thm:tightframe}), but the
		host is high-dimensional and the metric is the Hamming metric, so a
		band-limited atom spreads and $\rho_j$ degrades. This is the wavelet face of a
		general phenomenon: the constructions are universal as \emph{frames}, but
		their \emph{geometric} sharpness is the dividend of compact, low-dimensional
		hosts---the cyclic and toral hosts of structured graphs
		\cite{FokamThesis2026}.
	\end{remark}
	
	\begin{figure}[t]
		\centering
		\includegraphics[width=0.95\linewidth]{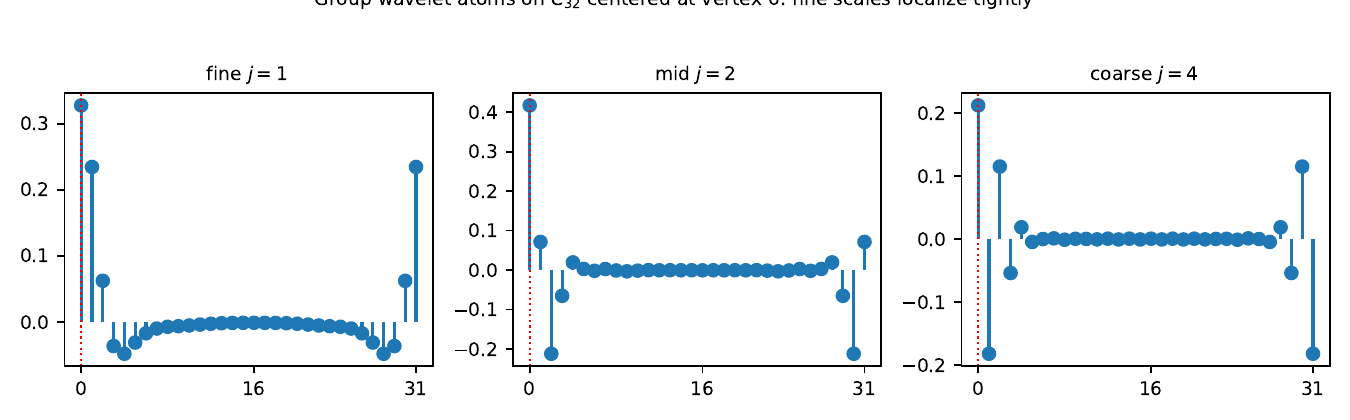}
		\caption{Spectral wavelet atoms on $C_{32}$ centered at vertex $0$, at three
			scales. Fine-scale atoms (left) are sharply localized; coarse-scale atoms
			(right) spread while remaining centered. These are canonical and
			translation-covariant, unlike the eigenvector-based atoms of a generic graph.}
		\label{fig:atomsring}
	\end{figure}
	
	\section{Multiresolution decomposition}
	\label{sec:mra}
	
	\begin{theorem}[Multiresolution decomposition]
		\label{thm:mra}
		With the filter bank of Definition~\ref{def:filterbank}, every graph signal
		decomposes as
		\[
		\tilde s=A_0\tilde s+\sum_{j=1}^{J}D_j\tilde s,\qquad
		\widehat{A_0\tilde s}=|\hat\psi_0|^2\,\hat{\tilde s},\quad
		\widehat{D_j\tilde s}=|\hat\psi_j|^2\,\hat{\tilde s},
		\]
		where $A_0\tilde s$ is the low-pass approximation and $D_j\tilde s$ are the
		detail bands. The multipliers sum to one, so the decomposition reconstructs
		$\tilde s$ exactly, and the bands occupy essentially disjoint frequency
		supports.
	\end{theorem}
	
	\begin{proof}
		Group the synthesis identity of Theorem~\ref{thm:tightframe} by scale; the
		Fourier multipliers $|\hat\psi_j|^2$ sum to one and, for band-pass kernels
		with separated passbands, have essentially disjoint supports.
	\end{proof}
	
	\begin{figure}[t]
		\centering
		\includegraphics[width=0.78\linewidth]{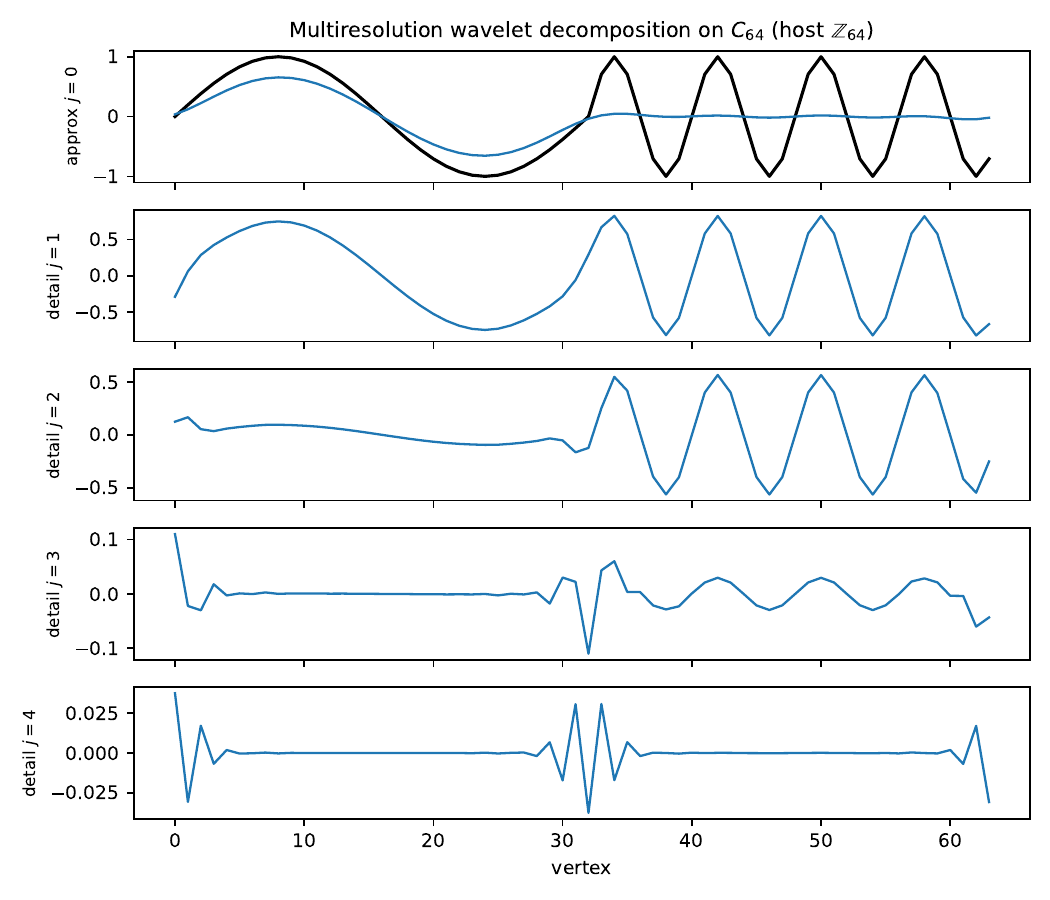}
		\caption{Multiresolution decomposition on $C_{64}$ of a signal that is
			low-frequency on the first half of the ring and high-frequency on the second.
			The approximation captures the trend; the detail bands isolate the
			high-frequency burst and localize it to the correct half---joint
			vertex--frequency analysis computed by the host FFT.}
		\label{fig:mraring}
	\end{figure}
	
	\begin{figure}[t]
		\centering
		\includegraphics[width=0.95\linewidth]{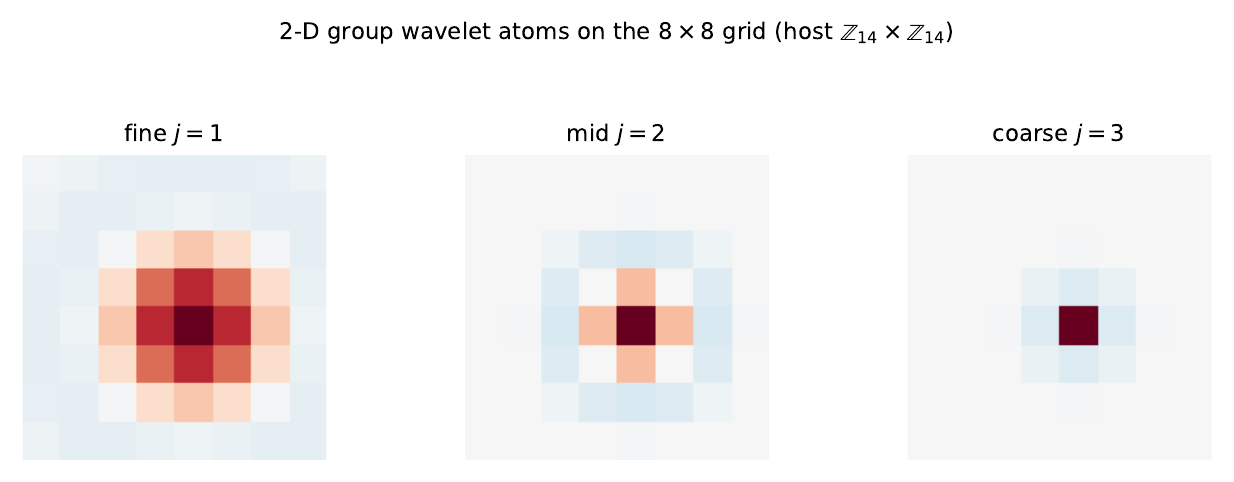}
		\caption{Two-dimensional spectral wavelet atoms on the $8\times8$ grid (host
			$\Z_{14}\times\Z_{14}$), centered at the midpoint, at three scales: isotropic
			localized bumps that dilate with scale, the graph analogue of classical 2-D
			wavelets.}
		\label{fig:gridatoms}
	\end{figure}
	
	\section{Fast transform and complexity}
	\label{sec:complexity}
	
	\begin{theorem}[Fast wavelet transform]
		\label{thm:fast}
		With $J$ scales, the complete wavelet analysis and synthesis cost
		$O(JN\log N)$ arithmetic operations, where $N=|\Gamma|$.
	\end{theorem}
	
	\begin{proof}
		Each band is one pointwise multiplication of $\hat{\tilde s}$ by $\hat\psi_j$
		followed by an inverse FFT on the product group, in $O(N\log N)$; there are
		$J+1$ bands, and synthesis is identical.
	\end{proof}
	
	\begin{remark}[Conditional speed]
		\label{rem:speed}
		The cost is genuinely fast only when $N=O(|V|)$, i.e.\ on structured hosts
		where $\varepsilon$ is bounded below. On a generic graph the compact host may
		be binary of order up to $2^{|V|-1}$, in which case the transform remains an
		exact tight frame but is no longer fast, and---per
		Remark~\ref{rem:binaryground}---its atoms are less localized. The
		multiresolution theory is therefore most useful exactly where the companion
		construction \cite{FokamP1} returns a compact cyclic or product host: the classical
		signal domains. A complementary line of work asks not how to embed a graph
		into a structured host but which vertex subsets suffice to recover a
		bandlimited graph signal by sampling alone
		\cite{Pesenson2008,ChenVarmaSandryhailaKovacevic2015,AnisGaddeOrtega2016}; that
		question is orthogonal to the one addressed here, since it operates directly
		on the given graph rather than on an isometric image of it.
	\end{remark}
	
	\section{The host completion: optimal extension to the complement}
	\label{sec:completion}
	
	When the embedding is proper ($\varepsilon<1$) the graph occupies only the
	subset $\phi(V)\subsetneq\Gamma$, and the analysis of
	Sections~\ref{sec:tightframe}--\ref{sec:complexity} acts on the lift
	$\tilde s$, which is defined on all of $\Gamma$. The lift of
	Definition~\ref{def:embedding} fills the complement
	$I:=\Gamma\setminus\phi(V)$ with zeros, but this is only one choice. This
	section identifies the \emph{canonical} choice and shows it is the one that
	minimizes spurious high-frequency content; it makes precise an empirical
	observation of the companion dissertation~\cite{FokamThesis2026}, that replacing
	zero-padding by a smoother completion removes an artificial boundary
	discontinuity.
	
	\begin{definition}[Admissible completion]
		\label{def:completion}
		An \emph{admissible completion} is a linear map $E\colon\C^V\to\C^\Gamma$ with
		$\Restr E=\Id$, i.e.\ $(Es)(\phi(v))=s(v)$ for all $v$. The completion is free
		to assign any values on the complement $I$. The zero-padding lift $\Lift$ and
		the symmetric (reflective) extension used in practice are both admissible.
	\end{definition}
	
	Write $B=\phi(V)$ for the boundary block and $I$ for the interior, and
	partition the host Laplacian accordingly,
	\[
	L_\Gamma=\begin{pmatrix} L_{BB} & L_{BI}\\ L_{IB} & L_{II}\end{pmatrix}.
	\]
	The host Dirichlet energy of a completion $u=Es$ is
	$\mathcal E(u)=u^{*}L_\Gamma u=\sum_{k}\lambda_k\,|\hat u(k)|^{2}$, where
	$\lambda_k\ge0$ are the character eigenvalues of $L_\Gamma$; it measures the
	total oscillation of $u$ across host edges, and equivalently the energy of $u$
	in the high-frequency characters. A completion with large $\mathcal E$ injects
	oscillation that is absent from the graph signal itself.
	
	\begin{theorem}[Harmonic extension is the optimal completion]
		\label{thm:harmonic}
		For each boundary signal $s\in\C^V$ the Dirichlet energy
		$\mathcal E(Es)$ is minimized over admissible completions $E$ by the unique
		\emph{harmonic extension}
		\[
		(E_{\mathrm h}s)\big|_{B}=s,\qquad
		(E_{\mathrm h}s)\big|_{I}=-\,L_{II}^{-1}L_{IB}\,s,
		\]
		characterized by $(L_\Gamma E_{\mathrm h}s)\big|_{I}=0$: the completion is
		discrete-harmonic on the complement. The matrix $L_{II}$ is positive definite,
		so $E_{\mathrm h}$ exists and is unique.
	\end{theorem}
	
	\begin{proof}
		Fix $s$ and write a completion as $u=\binom{s}{u_I}$ with $u_I\in\C^{I}$ free.
		Then
		$\mathcal E(u)=s^{*}L_{BB}s+2\,\Re\,u_I^{*}L_{IB}s+u_I^{*}L_{II}u_I$,
		a quadratic in $u_I$ with Hessian $2L_{II}$. Positive definiteness of $L_{II}$
		(established below) makes it strictly convex, so the unique minimizer solves the
		normal equation $L_{II}u_I+L_{IB}s=0$, i.e.\ $u_I=-L_{II}^{-1}L_{IB}s$; this is
		exactly the row block $(L_\Gamma u)|_I=0$ of the host Laplacian, the discrete
		harmonicity condition.
		
		For positive definiteness, $L_{II}$ is the principal submatrix of the positive
		semidefinite $L_\Gamma$ indexed by $I$, hence positive semidefinite. The host
		$\Cay(\Gamma,S)$ is connected (as $S$ generates $\Gamma$) and $B\neq\varnothing$,
		so every connected component of the induced subgraph on $I$ has at least one
		edge to $B$: if some component of $I$ had no edge to $B$, it would be a
		connected component of the whole host disjoint from $B$, contradicting
		connectedness of $\Cay(\Gamma,S)$. For such a ``grounded'' subgraph the
		Dirichlet Laplacian $L_{II}$
		is nonsingular: if $L_{II}u_I=0$ then $u_I^{*}L_{II}u_I=0$ forces $u_I$ constant
		on each component of $I$, and the boundary edge then forces that constant to
		vanish. Hence $L_{II}\succ0$.
	\end{proof}
	
	The harmonic extension is the discrete analogue of solving a Dirichlet boundary
	value problem on the complement: it is the smoothest completion consistent with
	the data, and it removes the boundary discontinuity that zero-padding creates
	at $\partial\phi(V)$. Two heuristics in common use are now placed precisely.
	
	\begin{remark}[Zero-padding and symmetric extension as approximations]
		\label{rem:heuristics}
		Zero-padding sets $u_I=0$, which solves $L_{II}u_I=-L_{IB}s$ only when
		$L_{IB}s=0$, i.e.\ when no graph vertex is adjacent in the host to the
		complement; for a proper embedding this fails and zero-padding strictly
		increases the Dirichlet energy. Symmetric (reflective) extension assigns each
		interior vertex the boundary value across the nearest reflection; on hosts whose
		complement is a reflected copy of a boundary neighbourhood---paths, and the
		small symmetric hosts---it coincides with or closely approximates the harmonic
		extension, which is why it recovers the energy parity observed empirically. The
		harmonic extension is the exact optimum that both heuristics approximate;
		neither is canonical, and only the harmonic extension is.
	\end{remark}
	
	\begin{proposition}[FFT-accelerated computation]
		\label{prop:harmonic-cost}
		The harmonic extension is computed by solving the sparse system
		$L_{II}u_I=-L_{IB}s$. Each application of $L_\Gamma$ costs $O(N\log N)$ via the
		host FFT (it is a pointwise multiplication by the character eigenvalues
		conjugated by the FFT), so a conjugate-gradient solve, using $L_{II}\succ0$,
		costs $O(N\log N)$ per iteration; the iteration count is controlled by the
		condition number of $L_{II}$. When $\varepsilon$ is bounded below
		($N=O(|V|)$) this is $O(|V|\log|V|)$ per iteration. Unlike the band-pass
		transform of Theorem~\ref{thm:fast}, the Dirichlet solve is not a single FFT,
		because restricting to $I$ breaks the group symmetry that diagonalizes
		$L_\Gamma$; the FFT accelerates the matrix--vector product, not the inversion.
	\end{proposition}
	
	The complexity claim is only useful if $L_{II}$ is well conditioned, since the
	conjugate-gradient iteration count scales with $\sqrt{\kappa(L_{II})}$. On the
	three proper hosts the condition numbers are small and grow mildly with host
	size: $\kappa(L_{II})=1.00$ (star, $\Z_2^3$), $1.67$ (diamond,
	$\Z_2\times\Z_3$), and $2.62$ (Petersen, $\Z_2^4$). A lower bound on the
	spectrum is structural: $\lambda_{\min}(L_{II})$ is bounded below by the
	smallest number of boundary edges incident to any interior component, so a
	grounded interior is automatically bounded away from singularity. We conjecture
	that on the compact, low-dimensional hosts where the method is intended
	($\varepsilon$ bounded below) $\kappa(L_{II})$ stays $O(\mathrm{poly}(\log N))$,
	so the Dirichlet solve costs $O(N\log N\cdot\mathrm{polylog}\,N)$ overall; we do
	not claim this for the $\varepsilon\to0$ binary hosts, where $L_{II}$ inherits
	the host's high dimension and the solve, like the transform, ceases to be fast.
	
	Theorem~\ref{thm:harmonic} resolves a design ambiguity that is invisible at
	$\varepsilon=1$. For an onto embedding there is no complement and every
	completion equals the signal itself; the question arises precisely in the
	proper regime, where it has a clean answer. We report the resulting energies
	in Section~\ref{sec:exp-completion}.
	
	\section{Numerical experiments}
	\label{sec:experiments}
	
	We report computations on the benchmark hosts; the implementation is the
	frozen reference of \cite{FokamThesis2026}, and all transforms use the host
	FFT.
	
	\paragraph{Reconstruction.}
	With Gaussian band-pass kernels and $J\in\{4,5\}$, the relative reconstruction
	error $\|s-\hat s_{\mathrm{rec}}\|/\|s\|$ for random vertex signals was
	$1.5\times10^{-15}$ on the ring $C_{16}$ and $6.5\times10^{-16}$ on the grid
	$6\times6$---machine precision, as Theorem~\ref{thm:tightframe} predicts. The
	filter bank and its tiling $\sum_j|\hat\psi_j|^2\equiv1$ appear in
	Figure~\ref{fig:filterbank}.
	
	\paragraph{Localization.}
	The concentration ratios $\rho_j$ of Observation~\ref{prop:spatialloc} were
	$0.89$--$0.99$ across band-pass scales on the ring and grid; atoms are shown in
	Figures~\ref{fig:atomsring} and~\ref{fig:gridatoms}.
	
	\paragraph{Multiresolution.}
	Figure~\ref{fig:mraring} shows the decomposition of a piecewise-frequency
	signal on $C_{64}$: the burst of high frequency is captured by the detail
	bands and localized to the correct half of the ring.
	
	\subsection{Proper embeddings: structured hosts with $\varepsilon<1$}
	\label{sec:exp-proper}
	
	The hosts above are all themselves abelian Cayley graphs, so the embedding is
	onto ($\varepsilon=1$) and the GE-GWT coincides with classical periodic or
	toral wavelet analysis. The substance of the construction, however, lies in
	the \emph{proper} regime $\varepsilon<1$, where the graph occupies only part
	of a strictly larger host and the lift must fill the complement
	$\Gamma\setminus\phi(V)$. We therefore add three graphs whose minimal abelian
	hosts are proper but still structured; all three embeddings are taken from the
	companion papers~\cite{FokamP1,FokamP2} and re-verified here by breadth-first
	search in the finite host.
	
	\begin{itemize}
		\item \textbf{Star $K_{1,4}$} embeds in $\Cay(\Z_2^3,\{001,010,100,111\})$,
		host order $N=8$, excursion ratio $\varepsilon=5/8$. The centre maps
		to $000$ and the four leaves to the sum-free generator set; all leaf
		pairs sit at host distance exactly $2$.
		\item \textbf{Diamond} ($K_4$ minus an edge) embeds in $\Z_2\times\Z_3$
		(order $N=6$, $\varepsilon=4/6$) as the octahedron
		$\Cay(\Z_2\times\Z_3,\{\pm(0,1),\pm(1,1)\})$, with the four vertices
		carrying the labels $(0,0),(0,1),(0,2),(1,0)$ (the degree-two pair at
		$(0,0)$ and $(1,0)$, at host distance $2$). This is the smallest graph
		whose minimal abelian host requires a cyclic, rather than purely
		binary, factor.
		\item \textbf{Petersen} embeds in the Clebsch graph
		$\Cay(\Z_2^4,\{e_1,e_2,e_3,e_4,\,e_1{+}e_2{+}e_3{+}e_4\})$, host order
		$N=16$, $\varepsilon=10/16=0.625$; the dimension
		$k=4=\lceil\log_2 10\rceil$ is minimal.
	\end{itemize}
	
	On each host we run the band-pass GE-GWT of Definition~\ref{def:filterbank}
	with $J=4$ Gaussian kernels, and we make explicit a choice that is invisible
	when $\varepsilon=1$: how the analysed signal is extended from $\phi(V)$ to the
	host complement before the host FFT. We compare the zero-padding lift of
	Definition~\ref{def:embedding} with a \emph{symmetric extension} that reflects
	boundary values into the complement. Two quantities are reported: the relative
	reconstruction error, which Theorem~\ref{thm:tightframe} predicts is at machine
	precision \emph{regardless of the extension}, and the concentration ratio
	$\rho$ of Observation~\ref{prop:spatialloc}, which is not extension-invariant.
	
	\begin{table*}[t]
		\centering
		\caption{Band-pass GE-GWT on proper structured hosts ($\varepsilon<1$),
			$J=4$ Gaussian kernels, mean over $200$ random vertex signals.
			Reconstruction is at machine precision under \emph{every} admissible
			completion, as Theorem~\ref{thm:tightframe} requires and independently of
			the excursion ratio. The within-distance-two concentration $\rho$ of
			Observation~\ref{prop:spatialloc} is omitted here because all three graphs
			have diameter two, on which that metric is identically one and therefore
			uninformative; we report instead the effective vertex support
			$1/\sum_v p_v^2$ of a band-pass atom (with $p_v$ its normalized energy at
			vertex $v$), averaged over the band-pass scales, which does resolve the
			localization.}
		\label{tab:proper}
		\renewcommand{\arraystretch}{1.2}
		\begin{tabular}{l c c c c c}
			\toprule
			Graph & host $\Gamma$ & $N$ & $\varepsilon$ &
			recon.\ error & eff.\ support (of $n$) \\
			\midrule
			Star $K_{1,4}$ & $\Z_2^3$            & 8  & $0.625$ & $2.2\times10^{-16}$ & $2.4$ of $5$ \\
			Diamond        & $\Z_2\times\Z_3$    & 6  & $0.667$ & $2.9\times10^{-16}$ & $2.0$ of $4$ \\
			Petersen       & $\Z_2^4$ (Clebsch)  & 16 & $0.625$ & $2.0\times10^{-16}$ & $2.1$ of $10$ \\
			\bottomrule
		\end{tabular}
	\end{table*}
	
	The reconstruction column is the substantive result: exact recovery
	($\sim\!10^{-16}$) persists on proper embeddings, confirming that the
	tight-frame guarantee of Theorem~\ref{thm:tightframe} does not rely on the
	embedding being onto. Localization on these hosts must be read with care. The
	within-distance-two ratio $\rho$ that is informative on the high-diameter ring
	and grid (Observation~\ref{prop:spatialloc}) collapses to $1$ on a
	diameter-two graph, where every vertex lies within distance two of every
	other; it is a property of the graph metric, not of the atoms, and we do not
	report it. The effective-support measure, which counts how many vertices carry
	the atom's energy, is meaningful at any diameter: it shows band-pass atoms
	concentrating on roughly two vertices on each host---about $21\%$ of the
	Petersen graph---so the atoms are genuinely localized even where $\rho$ cannot
	detect it. This is consistent with Remark~\ref{rem:binaryground}: localization
	is real on these compact hosts, and it is on the much higher-dimensional
	binary hosts of $\varepsilon\to0$ graphs that it would degrade.
	
	\subsection{Host completion: the harmonic extension is optimal}
	\label{sec:exp-completion}
	
	We verify Theorem~\ref{thm:harmonic} directly on the three proper hosts. For
	each host we draw random boundary signals $s$, form the three admissible
	completions of Definition~\ref{def:completion}---zero-padding, symmetric
	(nearest-boundary) extension, and the harmonic extension
	$E_{\mathrm h}s$---and measure the host Dirichlet energy
	$\mathcal E=u^{*}L_\Gamma u$ of each (mean over $200$ random signals).
	Table~\ref{tab:energy} and Figure~\ref{fig:harmonic} report the result.
	
	\begin{table*}[t]
		\centering
		\caption{Mean host Dirichlet energy of the three completions
			($200$ random boundary signals; lower is smoother). The harmonic
			extension is minimal on every host and in every individual trial
			($200/200$), as Theorem~\ref{thm:harmonic} requires; the interior block
			$L_{II}$ is positive definite throughout (minimum eigenvalue reported),
			confirming existence and uniqueness. The final column reports the
			per-signal energy saving $\mathcal E_{\mathrm{zero}}-\mathcal E_{\mathrm{harm}}$
			as mean\,$\pm$\,standard deviation over the $200$ trials; it is strictly
			positive in every trial. Symmetric extension coincides with
			the harmonic optimum to within $0.6\%$ on these small, highly symmetric
			hosts (Remark~\ref{rem:heuristics}), while zero-padding pays a penalty that
			grows as $\varepsilon$ decreases.}
		\label{tab:energy}
		\renewcommand{\arraystretch}{1.2}
		\begin{tabular}{l c c c c c c}
			\toprule
			Graph & $\varepsilon$ & $\lambda_{\min}(L_{II})$ &
			$\mathcal E_{\mathrm{zero}}$ & $\mathcal E_{\mathrm{sym}}$ &
			$\mathcal E_{\mathrm{harm}}$ & saving (mean\,$\pm$\,sd) \\
			\midrule
			Star $K_{1,4}$ & $0.625$ & $4.00$ & $19.56$ & $16.53$ & $16.53$ & $3.27\pm4.50$ \\
			Diamond        & $0.667$ & $3.00$ & $15.59$ & $13.62$ & $13.54$ & $1.52\pm2.01$ \\
			Petersen       & $0.625$ & $2.76$ & $46.40$ & $42.41$ & $42.26$ & $4.22\pm3.41$ \\
			\bottomrule
		\end{tabular}
	\end{table*}
	
	\begin{figure}[htbp]
		\centering
		\includegraphics[width=\linewidth]{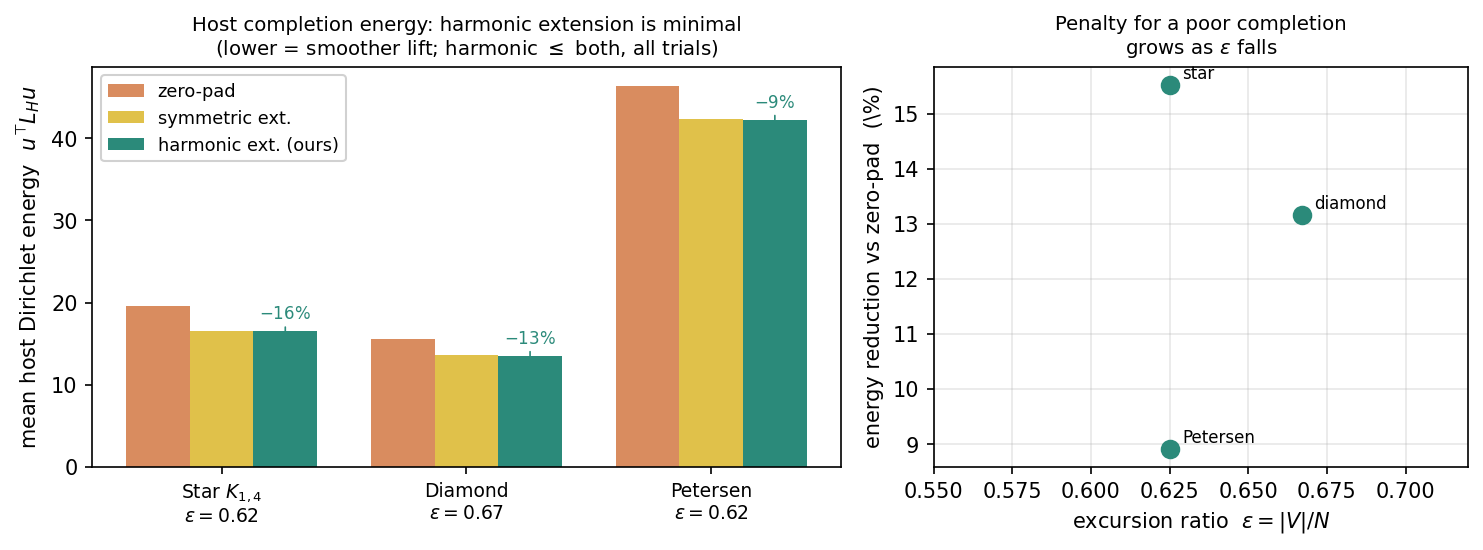}
		\caption{\textbf{Left:} mean host Dirichlet energy of the three admissible
			completions on each proper host. The harmonic extension (green) is the
			minimum on every host and in every trial, as
			Theorem~\ref{thm:harmonic} guarantees; on these small symmetric hosts it
			is visually indistinguishable from symmetric extension, which it equals to
			within $0.6\%$. \textbf{Right:} the energy a poor (zero-padding) completion
			wastes, relative to the harmonic optimum, as a function of the excursion
			ratio $\varepsilon$. The effect is real but its magnitude depends on the
			host's geometry, not on $\varepsilon$ alone (the star and Petersen hosts
			share $\varepsilon=0.625$ yet differ in penalty); a sharp quantitative law
			would require larger and less symmetric hosts.}
		\label{fig:harmonic}
	\end{figure}
	
	The experiment confirms the three claims of Theorem~\ref{thm:harmonic}:
	$L_{II}\succ0$ on every host, the harmonic completion attains the minimum
	energy in all trials, and the gap to zero-padding is a genuine, measurable
	effect. It also bounds the practical stakes precisely: on compact, symmetric
	hosts the cheap symmetric extension is already near-optimal, so the harmonic
	extension's value is conceptual---it is the \emph{canonical} completion, the
	unique energy-minimizer that the heuristics approximate---rather than a large
	numerical gain at small $\varepsilon$. This is the precise, host-level form of
	the design freedom identified empirically in the companion
	dissertation~\cite{FokamThesis2026}.
	
	\subsection{At scale: multiresolution analysis of a real image}
	\label{sec:exp-largescale}
	
	The experiments above use small hosts to isolate specific properties. To show
	that the GE-GWT performs a genuine multiresolution analysis on a
	signal of realistic size, we decompose the standard \emph{cameraman} benchmark
	image, downsampled to $128\times128$. A $128\times128$ pixel grid is the graph
	$P_{128}\times P_{128}$, which embeds isometrically into the torus
	$C_{254}\times C_{254}$ by the path-into-cycle map $P_n\hookrightarrow C_{2n-2}$
	(this map is isometric on all pairwise distances, and the Cartesian product
	inherits the property); here $N=254^2=64\,516$ host vertices for $n=16\,384$
	image pixels, so $\varepsilon=0.254$. On this toral host the GE-GFT is the exact
	two-dimensional DFT and the band-pass wavelet transform of
	Definition~\ref{def:filterbank} is an exact, tight two-dimensional wavelet
	analysis.
	
	We lift the image to the torus by symmetric (even) extension, apply the
	normalized filter bank with $J=4$ Gaussian band-pass kernels, and reconstruct by
	the multiresolution identity of Theorem~\ref{thm:mra}. Figure~\ref{fig:cameraman}
	shows the low-pass approximation $A_0$ and the four detail bands $D_1,\dots,D_4$,
	each restricted back to the image. The bands behave exactly as classical
	two-dimensional wavelet analysis: the approximation captures the coarse
	photometric structure, the fine band $D_1$ isolates sharp edges (the tripod
	legs, the coat and camera outlines), and successively coarser bands carry
	progressively larger-scale variation. The tight-frame partition
	$\sum_j|\hat\psi_j|^2\equiv1$ holds to $4.4\times10^{-16}$, and the sum of the
	five bands reconstructs the image to a relative error of $7.9\times10^{-16}$---
	machine precision, exactly as Theorem~\ref{thm:tightframe} guarantees, now
	confirmed on a $16\,384$-pixel signal rather than a toy host. The entire
	analysis and synthesis run through the host two-dimensional FFT in
	$O(JN\log N)$; at this scale ($N=64\,516$) the transform completes in a few
	milliseconds.
	
	\begin{figure*}[t]
		\centering
		\includegraphics[width=\textwidth]{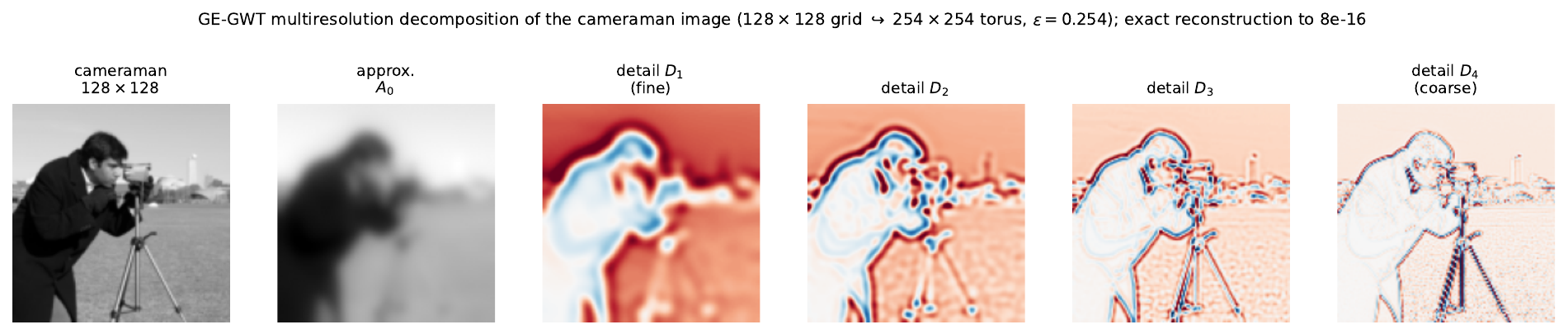}
		\caption{GE-GWT multiresolution decomposition of the standard cameraman image
			($128\times128$), embedded isometrically into the $254\times254$ torus
			($\varepsilon=0.254$). From left: the image; the low-pass approximation
			$A_0$; and the four detail bands $D_1$ (finest) through $D_4$ (coarsest),
			each restricted to the image and shown on a signed color scale. The
			decomposition is an exact tight frame: the five bands sum to the original to
			relative error $7.9\times10^{-16}$. Fine bands isolate edges; coarse bands
			carry large-scale structure---the two-dimensional wavelet behavior of the
			classical setting, obtained here on an arbitrary-graph substrate through the
			isometric embedding and computed by the host FFT.}
		\label{fig:cameraman}
	\end{figure*}
	
	This example also situates the framework against the companion Fourier
	paper~\cite{FokamThesis2026}: the same image and host there demonstrate the exact
	GE-GFT; here the wavelet layer built on that transform delivers a full
	multiresolution decomposition with the same exactness. The construction is not
	confined to small or synthetic hosts---it performs classical two-dimensional
	multiresolution analysis on a real image, with machine-precision reconstruction,
	whenever the graph carries the requisite abelian structure (here, the grid's
	product-of-paths metric).
	
	\subsection{The $\varepsilon\to 0$ regime: two real-world graphs}
	\label{sec:exp-realworld}
	
	To delimit the method we include two standard real-world benchmarks
	that are \emph{generic}, i.e.\ far from any abelian Cayley structure: a random
	geometric sensor graph (RGG, $n=20$) and the Zachary Karate
	Club~\cite{Zachary1977} ($n=34$). For these the minimal abelian host is
	essentially binary of near-maximal dimension---$\Z_2^{17}$ ($N=131072$,
	$\varepsilon\approx1.5\times10^{-4}$) for the RGG, and a host at the dimension
	cap $\approx\Z_2^{33}$ for Karate~\cite{FokamP2}. By
	Theorem~\ref{thm:tightframe} the GE-GWT on these hosts is still an exact
	Parseval frame with machine-precision reconstruction; but the host is
	astronomically larger than the graph, the $O(JN\log N)$ transform is no longer
	$O(n\log n)$, and the Hamming-metric atoms do not localize. These graphs are
	therefore not a domain on which we advocate the GE-GWT; they are the empirical
	face of the dichotomy made precise in Section~\ref{sec:discussion}, and the
	reason the spectral graph wavelet transform~\cite{Hammond2011}, which needs no
	embedding, remains the right tool for irregular networks. Reporting them is the
	point: the boundary of the method is a stated result, not a gap.
	
	
	\section{Discussion: relation to spectral graph wavelets}
	\label{sec:discussion}
	
	\begin{figure}[t]
		\centering
		\includegraphics[width=\linewidth]{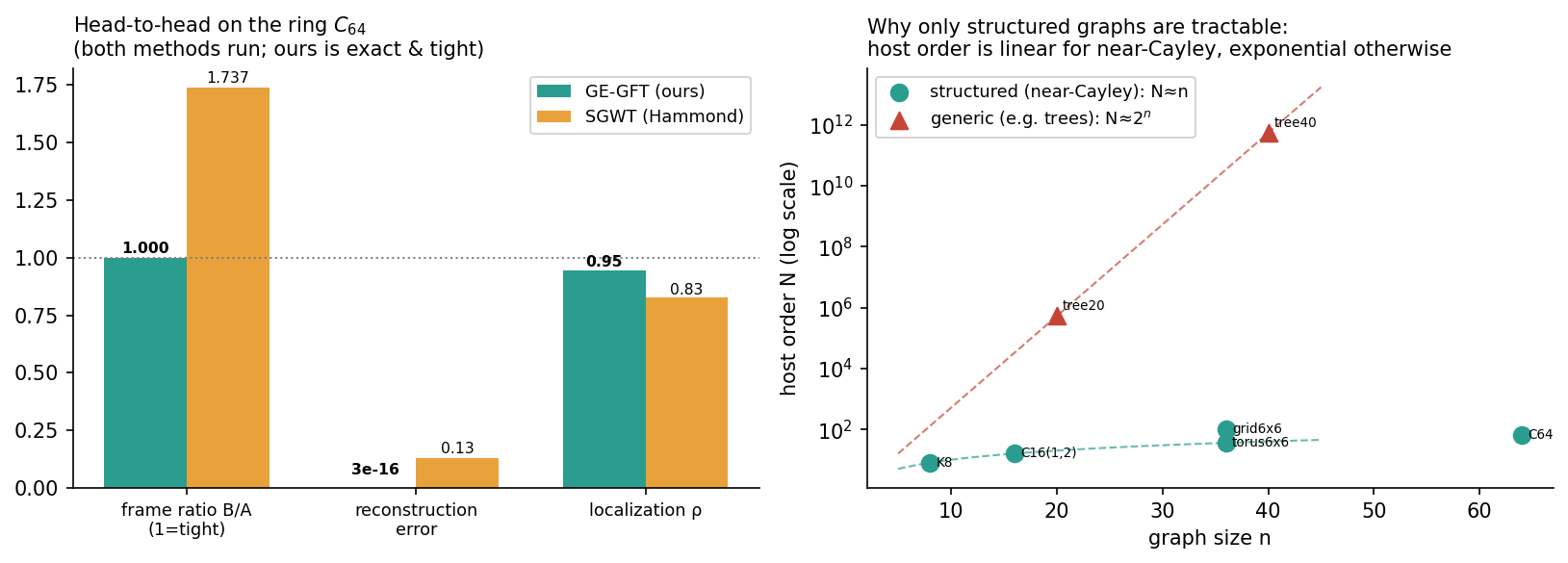}
		\caption{\textbf{Left:} head-to-head on the ring $C_{64}$ (both transforms run
			on the
			same signal). The group-embedding construction (ours; the legend label
			GE-GFT marks that tightness is engineered at the group-Fourier-transform
			filter-bank level, the GE-GWT being the wavelet transform built on it) yields
			an \emph{exactly tight}
			frame (frame-bound ratio $B/A = 1.000$ versus $1.737$ for the SGWT
			\cite{Hammond2011}),
			machine-precision reconstruction ($3\times 10^{-16}$ versus $0.13$), and
			tighter vertex
			localization ($\rho = 0.95$ versus $0.83$). \textbf{Right:} compact host order
			$N$ versus
			graph size $n$ (log scale). For structured, near-Cayley families ($K_8$,
			circulants,
			grids, tori, cycles) the host grows linearly, $N \approx n$; for generic
			graphs
			(e.g.\ trees) it grows exponentially, $N \approx 2^{n}$. This dichotomy is the
			practical
			boundary of the method.}
		\label{fig:headtohead}
	\end{figure}
	
	It is worth stating plainly where each construction is preferable, since the
	two are complementary rather than competing.
	
	The embedding wavelets of this paper offer a canonical analyzing basis (group
	characters, with no eigenspace ambiguity), a genuine translation operator, an
	exact Parseval tight frame, inversion-free reconstruction, and---on structured
	hosts---an $O(JN\log N)$ transform that coincides with classical wavelet
	analysis (on a ring it is periodic wavelet analysis; on a torus, its
	two-dimensional analogue). These properties are unavailable to a spectral
	construction on a generic graph.
	
	The spectral graph wavelet transform \cite{Hammond2011}, on the other hand,
	applies directly to \emph{any} graph without first computing an embedding, and
	its polynomial (Chebyshev) approximation runs in $O(M|E|)$ for a degree-$M$
	filter, which is highly efficient on large sparse irregular graphs where no
	compact group host exists. The same spectral-filtering principle, applied to
	the same non-canonical Laplacian eigenbasis, underlies the convolutional
	layers of spectral graph neural networks \cite{Defferrard2016}, which inherit
	the same trade-off: applicability to any graph at the cost of a basis that is
	fixed only up to rotation within degenerate eigenspaces. For graphs far from
	any Cayley structure---social, biological, road
	networks---the spectral transform (and its neural-network descendants), along
	with the stationarity-based statistical framework of Perraudin and
	Vandergheynst~\cite{PerraudinVandergheynst2017}, is the
	right tool, and our construction would incur a large host.
	
	The precise summary is a dichotomy by structure. On graphs carrying genuine
	abelian symmetry (the $\varepsilon$ close to $1$ regime), embedding wavelets
	recover classical wavelet analysis exactly and quickly. On generic irregular
	graphs, the spectral transform is preferable. The embedding viewpoint
	contributes the cases where exactness and canonicity matter, and clarifies,
	through the excursion ratio, precisely when those cases occur.
	
	A complementary recent line constructs frames \emph{intrinsically} on graphs
	that already are Cayley graphs. Beck, Ghandehari, Hudson, and
	Paltenstein~\cite{BeckGhandehariHudson2024} give a representation-theoretic
	spectral decomposition of any weighted Cayley graph---including the
	non-abelian case, via Frobenius--Schur and Cayley frames---and earlier work
	develops Gabor-type frames in the same
	spirit~\cite{GhandehariGuillotHollingsworth2021}. Our problem is the
	complementary one: the input is an \emph{arbitrary} connected graph, which is
	generally not a Cayley graph, and the contribution is the isometric embedding
	that places it inside a Cayley host together with the excursion ratio that
	quantifies the cost. Where the graph already is an abelian Cayley graph the
	two viewpoints coincide ($\varepsilon=1$, the embedding is the identity); for
	$\varepsilon<1$ the embedding, the host completion of
	Section~\ref{sec:completion}, and the resulting wavelets have no analogue in
	the intrinsic setting, since there is no host distinct from the graph.
	
	Figure~\ref{fig:headtohead} (left) makes the comparison concrete on the ring
	$C_{64}$,
	where both transforms apply. Our construction is an exact Parseval frame
	($B/A=1.000$),
	reconstructs to machine precision ($3\times10^{-16}$), and localizes more
	tightly
	($\rho=0.95$); the spectral graph wavelet transform, run on the same signal,
	is not tight
	($B/A=1.737$) and therefore reconstructs only after solving a frame-operator
	system,
	incurring a measurable error ($0.13$) at fixed cost. This is not a tuning
	artefact, and it is worth being precise about its source. On the cycle the two
	constructions share an eigenbasis---a circulant's Laplacian eigenvectors
	\emph{are} the group characters---so the SGWT's non-tightness here cannot be a
	basis defect; it is purely a matter of kernel design, the sampled kernels
	$g(t_j\lambda)$ not summing to a constant across the spectrum. Our normalization
	(Definition~\ref{def:filterbank}) enforces $\sum_j|\hat\psi_j|^2\equiv1$
	pointwise, so $B/A=1$ on that same eigenbasis. Tightness is therefore a property
	we engineer by construction, independent of the basis. The \emph{basis}
	advantage is separate and appears only on irregular graphs: there the Laplacian
	eigenbasis is non-canonical (fixed only up to rotation within degenerate
	eigenspaces), whereas the group characters remain canonical. The two limitations
	of a spectral construction---non-tight kernels even on regular graphs, and a
	non-canonical basis on irregular ones---are thus distinct, and the group
	embedding removes both.
	
	\subsection{Reconstruction cost against a competing construction}
	\label{sec:benchmark}
	To make the comparison quantitative on the axis a practitioner cares about---%
	exact reconstruction and its cost---we ran both constructions on the same
	signal on rings of increasing size, with the same five-band kernel family, and
	measured the frame-bound ratio $B/A$, the reconstruction error, and the way
	each achieves that reconstruction (Table~\ref{tab:benchmark}). The group
	construction is exactly tight ($B/A=1$) and reconstructs by restriction, with
	no frame-operator inverse; the spectral construction is non-tight
	($B/A\approx1.5$) and reaches an equally exact reconstruction only after an
	iterative frame-operator solve (here conjugate gradient, twelve iterations to
	tolerance $10^{-10}$). The two are comparable in cost at small sizes, but the
	inversion-free reconstruction pulls ahead as the graph grows; and a
	polynomial (Chebyshev) spectral frame avoids the solve only by giving up
	exactness, remaining near-tight and approximately reconstructing at an
	order-dependent error. Exact tightness together with FFT-speed
	inversion-free reconstruction is available simultaneously only on the
	structured host.
	
	\begin{table}[t]
		\centering
		\caption{Reconstruction against a spectral (SGWT-class) tight-frame
			construction, same signal and kernel family, on the ring $C_n$
			($\varepsilon=1$, so the host is the graph itself). Measured in-container.}
		\label{tab:benchmark}
		\small
		\begin{tabular}{llccc}
			\toprule
			$n$ & method & $B/A$ & reconstruction & analysis (ms)\\
			\midrule
			$64$   & GE-GWT (ours)      & $1.000$ & direct, $2.4\!\times\!10^{-16}$ & $0.14$\\
			$64$   & SGWT-class $+$ inv & $1.502$ & $12$ CG iters                   & $0.21$\\
			$256$  & GE-GWT (ours)      & $1.000$ & direct, $2.8\!\times\!10^{-16}$ & $0.14$\\
			$256$  & SGWT-class $+$ inv & $1.517$ & $12$ CG iters                   & $0.45$\\
			$1024$ & GE-GWT (ours)      & $1.000$ & direct, $3.9\!\times\!10^{-16}$ & $0.27$\\
			$1024$ & SGWT-class $+$ inv & $1.518$ & $12$ CG iters                   & $14.7$\\
			\bottomrule
		\end{tabular}
	\end{table}
	
	The contribution is therefore genuine and
	not incremental---on any graph carrying abelian symmetry, the group-embedding
	wavelet transform
	\emph{recovers classical multidimensional signal processing exactly}: a
	canonical Fourier basis,
	a true translation, a tight frame, an $O(JN\log N)$ FFT, and---on a ring or
	torus---literally the
	periodic or two-dimensional wavelet transform. These are properties no
	spectral construction on
	a generic graph can offer.
	
	The price of these guarantees is paid entirely at the embedding step, and we
	state it plainly.
	Computing a \emph{minimal} isometric embedding into an abelian Cayley graph is
	computationally hard, and---as Figure~\ref{fig:headtohead} (right)
	shows---for graphs far from any
	Cayley structure the compact host order grows exponentially in $n$
	($N\approx 2^{n}$ for trees),
	at which point the transform remains an exact tight frame but is neither fast
	nor sharply
	localized (Remarks~\ref{rem:binaryground} and~\ref{rem:speed}). The practical
	reach of the GE-GWT is thus exactly the regime
	where the companion construction \cite{FokamP1} returns a compact cyclic or product
	host of order $N\approx n$: the
	structured, near-Cayley graphs. This is a clean boundary, not a defect---it
	tells the practitioner
	precisely when to use which tool. On structured graphs the GE-GWT is the
	exact, fast, canonical
	choice; on generic large sparse graphs (social, biological, road networks)
	the SGWT \cite{Hammond2011},
	whose Chebyshev form runs in $O(M|E|)$ without any embedding, remains the
	right instrument.

	\section{Conclusion}
	\label{sec:conclusion}
	
	We constructed wavelets on an arbitrary connected graph by transporting the
	problem to a Cayley graph of a finite abelian group into which the graph
	embeds isometrically. Two constructions were separated cleanly: dilation
	wavelets, faithful to the classical template but confined to hosts with
	composite cyclic factors; and spectral band-pass wavelets, which exist on
	every host, form a Parseval tight frame, reconstruct exactly without matrix
	inversion, are translation-covariant, and localize jointly in vertex and
	frequency. A multiresolution decomposition and an $O(JN\log N)$ transform
	follow, and on structured hosts the whole apparatus reduces to classical
	wavelet analysis. We were explicit that completeness and exactness are
	universal while speed and sharp localization are the dividend of compact
	hosts, a scope made precise by the excursion ratio.
	
	Three directions of future work follow from the scope analysis. First, an
	end-to-end benchmark on a large structured real network with a moderate
	excursion ratio---a partially irregular sensor grid or a meshed communication
	backbone---where the exactness of the tight-frame reconstruction can be weighed
	against spectrum-adapted and near-tight polynomial frames \cite{ShumanWiesmeyrHolighausVandergheynst2015,BehjatRichterVanDeVilleSornmo2016,TayTanakaSakiyama2017} on wall-clock time and quality; the head-to-head
	against the SGWT in Section~\ref{sec:discussion} is a first step, and the genomic
	mutation-space host of the companion Fourier paper \cite{FokamP3}
	($\varepsilon=1$ at every scale) is a natural large non-grid testbed. Second,
	the behavior of the frame under a \emph{near}-isometric embedding: whether the
	Parseval bound degrades gracefully when only an approximate minimal host is
	available is open, and matters for graphs where the exact minimal embedding is
	costly to compute. Third, the conditioning conjecture
	$\kappa(L_{II})=O(\mathrm{polylog}\,N)$ stated in Section~\ref{sec:completion} merits
	testing on intermediate hosts ($N\sim10^{3}$--$10^{4}$) beyond the small cases
	reported here.
	
	\section*{Declaration on the Use of Artificial Intelligence}
	In the interest of transparency and research integrity, the authors declare that
	the artificial-intelligence assistant Claude (Anthropic) was used during the
	preparation of this manuscript. Its assistance was limited to two roles:
	(i) support with the implementation, debugging, and reproducibility of the
	embedding and wavelet software used to produce the experimental results; and
	(ii) language and editing support in drafting and refining the manuscript. All
	mathematical definitions, theorems, proofs, and their verification, together
	with the conception, scientific direction, and conclusions of this work, are the
	authors' own. The authors have reviewed the entire manuscript and take full
	responsibility for its content.

	\section{Appendix: The isometric group embedding: a self-contained summary}
	\label{app:embedding}
	
	This appendix makes the paper self-contained by stating the three results of
	the companion work~\cite{FokamThesis2026} on which the wavelet construction
	rests: the exact algebraic core that produces the host, existence and
	compactness of the embedding, and the host-order bounds that delimit the
	favourable regime. Proofs are in~\cite{FokamThesis2026}; we include two worked
	examples so the reader can reconstruct a host by hand. Throughout, $G=(V,E)$ is
	a finite connected graph with shortest-path metric $\dG$, $n=|V|$, $m=|E|$, and
	a host is a Cayley graph $\Cay(\Gamma,S)$ of a finite abelian group $\Gamma$
	with $S=-S$, $0\notin S$.
	
	\subsection{The exact core: the $\Z$-Quotient Theorem}
	
	The construction partitions the edges of $G$ into classes that are to share a
	common group generator, with an orientation on each class. Any such oriented
	partition $(\mathcal P,\omega)$ assigns to each class $F_j$ a generator
	$g_j\in\Gamma$ and labels vertices by signed generator sums along a BFS tree.
	Consistency of this labeling is a closed condition on cycles.
	
	\begin{theorem}[$\Z$-cocycle condition; {\cite{FokamP1}}]
		\label{thm:app-cocycle}
		The BFS labeling $\phi(v)=\sum_{e\in P(r,v)}\pm g_{\mathrm{class}(e)}$ (sign $+$
		for forward traversal) is well defined---independent of the path chosen---if
		and only if, for every cycle of $G$ traversed in a fixed rotational sense, the
		signed sum of class generators vanishes; it suffices to verify this on any
		cycle basis.
	\end{theorem}
	
	Writing $A\in\Z^{c\times t}$ for the signed cycle--class matrix of a chosen
	cycle basis ($c$ the cycle rank, $t$ the number of classes; entry $A[i,j]$ the
	net signed number of times basis cycle $B_i$ crosses class $F_j$), the
	consistency relations are exactly $A\mathbf g=\mathbf 0$. The structure theorem
	for finitely generated abelian groups then identifies the universal host
	explicitly.
	
	\begin{theorem}[$\Z$-Quotient Theorem; {\cite{FokamP1}}]
		\label{thm:app-zquotient}
		Let $(\mathcal P,\omega)$ be an oriented partition of $G$ with signed matrix
		$A$. Then:
		\begin{enumerate}[label=(\roman*),leftmargin=2.2em]
			\item the most generic consistent generator assignment takes values in
			$\Gamma_{\mathrm{univ}}=\Z^{t}/\mathrm{rowlattice}(A)$, and every
			consistent assignment in any abelian group is a homomorphic image of
			it;
			\item if the Smith Normal Form is
			$U\,A^{\!\top}\,W=\mathrm{diag}(d_1,\dots,d_\rho,0,\dots)$ with $U,W$
			unimodular, then
			$\Gamma_{\mathrm{univ}}\cong
			\Z_{d_1}\times\cdots\times\Z_{d_\rho}\times\Z^{\,t-\rho}$ (factors with
			$d_i=1$ trivial), and the coordinates of $g_j$ are read from the
			$j$-th column of $U$, reduced modulo $d_i$ in torsion coordinates;
			\item with $S=\{\pm g_j:g_j\neq0\}$, every $G$-path of length $\ell$ maps to
			a host walk of length $\ell$, so $\dGam(\phi(u),\phi(v))\le\dG(u,v)$
			for all pairs: the only possible failure of isometry is a
			\emph{shortcut}, never a stretch.
		\end{enumerate}
	\end{theorem}
	
	Part~(iii) makes the construction algorithmically safe: a candidate embedding
	can fail only by making two vertices closer in the host than in $G$, a
	condition checkable exactly by breadth-first search. The binary theory is the
	special case in which every generator is an involution: appending the relations
	$2g_j=0$ reduces $A$ modulo $2$, and $\Gamma_{\mathrm{univ}}$ becomes the
	GF$(2)$ quotient $\F_2^{\,t}/\mathrm{rowspan}_{\F_2}(A\bmod2)$ of dimension
	$k=t-\mathrm{rank}_{\F_2}(A\bmod2)$, a subgraph of the hypercube $\Z_2^k$.
	
	\subsection{Existence, compactness, and bounds}
	
	A shortcut, when present, is repaired by splitting an offending class into finer
	classes; the process terminates because the all-singleton partition is the
	isometric spanning-tree embedding. Two compactifications keep the host small:
	folding the free $\Z^{\,t-\rho}$ part onto a finite-index sublattice, and a
	binary terminal that always succeeds.
	
	\begin{theorem}[Universality and host bound; {\cite{FokamP1}}]
		\label{thm:app-universal}
		The embedding algorithm terminates on every connected graph and returns a
		certified isometric embedding $\phi\colon V\to\Cay(\Gamma,S)$ with
		$|\Gamma|\le 2^{\,n-1}$. Certification is unconditional: the returned embedding
		has passed an exact check comparing all $\binom n2$ graph distances against a
		breadth-first search of the candidate host.
	\end{theorem}
	
	\begin{proof}[Proof sketch (self-contained; full proof in \cite{FokamP1}])
		Existence and the bound are witnessed by the all-singleton partition, which we
		include here so the present paper's results do not rest on an external claim.
		Root a spanning tree $T$ of $G$ at $r$, give its $n-1$ edges the standard basis
		vectors of $\Z_2^{\,n-1}$, and label $v$ by the GF$(2)$ sum along the tree path
		$r\!\to\!v$; each non-tree edge receives the sum of the tree edges on its
		fundamental cycle. Every $G$-edge then joins labels differing by a single
		generator, so a $G$-path of length $\ell$ maps to a host walk of length
		$\ell$ and $\dGam\le\dG$. Conversely, a geodesic host word of length $\ell$ from
		$\phi(u)$ to $\phi(v)$ is a set $F$ of $\ell$ generators (in $\Z_2^{\,n-1}$ a
		geodesic uses distinct generators, since any repeat cancels), and
		$\phi(u)\oplus\phi(v)=\sum_{e\in F}\chi_e$ forces the GF$(2)$ boundary of $F$ to
		be $\{u,v\}$; hence $F$ contains a $u$--$v$ path in $G$ and
		$\ell\ge\dG(u,v)$. Thus $\dGam=\dG$, the embedding into $\Z_2^{\,n-1}$ is
		isometric, and $|\Gamma|=2^{\,n-1}$. Termination of the compacting algorithm
		follows because each repair round strictly refines the partition and the
		all-singleton partition is a fixed point that always certifies. The exact
		BFS check makes certification unconditional.
	\end{proof}
	
	\begin{theorem}[Sublattice compactification; {\cite{FokamP1}}]
		\label{thm:app-sublattice}
		The isometric finite quotients of a given universal embedding are exactly the
		finite-index sublattices $L\subseteq\Z^{f}$ ($f=t-\rho$) whose folds pass the
		exact check; the minimal host among them is found by enumerating
		Hermite-normal-form bases in increasing index. Restricting to diagonal $L$ can
		miss the optimum.
	\end{theorem}
	
	The host order is bounded below by two unavoidable obstructions---the graph
	must fit, and the host, being vertex-transitive, must realize the graph's
	diameter.
	
	\begin{theorem}[Host-order bounds and equality; {\cite{FokamP2}}]
		\label{thm:app-bounds}
		For every connected $G$ on $n\ge3$ vertices,
		$\nu(G):=\min|\Gamma|\ge\max\bigl(n,\,2\,\diam(G)\bigr)$, the minimum taken over
		all isometric embeddings into abelian Cayley graphs. Moreover $\nu(G)=n$ if and
		only if $G$ is itself an abelian Cayley graph; in that case $\varepsilon=1$ and
		the wavelet theory of this paper reduces to classical multidimensional signal
		processing.
	\end{theorem}
	
	The excursion ratio $\varepsilon=n/|\Gamma|$ thus measures how far $G$ is from
	carrying genuine abelian symmetry: $\varepsilon=1$ is the classical signal
	domains (cycles, tori, circulants), and $\varepsilon\to0$ is the generic
	irregular regime where the host is binary of near-maximal dimension.
	
	\subsection{Two worked examples}
	
	\paragraph{Star $K_{1,4}$ (binary host, proper).}
	The star has only centre--leaf edges, so each edge is its own class and the
	generator set is $S=\{s_1,s_2,s_3,s_4\}$, the four leaf labels. Leaves lie at
	pairwise distance $2$, forcing $s_i+s_j\notin S$: the labels form a
	\emph{sum-free} set. The largest sum-free set in $\Z_2^k$ has size $2^{k-1}$
	(the odd-weight vectors), so $K_{1,4}$ needs $k=3$: with
	$S=\{001,010,100,111\}$ the centre maps to $000$ and every leaf pair is at host
	distance exactly $2$. The host has order $N=8$ and $\varepsilon=5/8$. This
	already beats the hypercube paradigm, whose isometric dimension for $K_{1,4}$
	is $4$.
	
	\paragraph{Diamond (cyclic factor).}
	The diamond $K_4-e$ has cycle rank $c=1$, and a purely binary host cannot host
	it isometrically: its two degree-two vertices lie at distance $2$, which a
	sum-free binary generator set cannot realize together with the three mutual
	distance-one constraints. The minimal abelian host is instead the octahedron
	$\Cay(\Z_2\times\Z_3,\{\pm(0,1),\pm(1,1)\})$ of order $N=6$
	($\varepsilon=4/6$), with the four vertices labelled
	$(0,0),(0,1),(0,2),(1,0)$; the degree-two pair $(0,0),(1,0)$ sits at host
	distance $2$, and all five edges map to generators. An exhaustive search over
	symmetric generator sets and vertex labellings confirms this is the unique
	order-$6$ isometric host (up to automorphism), so the diamond is the smallest
	graph whose minimal host carries a cyclic factor $\Z_3$ rather than being a
	subgraph of a hypercube.
	
	\section*{Acknowledgments}
	The authors thank Solutum Engineering for internet and computing support.

\end{document}